\documentclass[journal]{IEEEtran}
\usepackage{amsmath,amssymb,amsfonts}
\usepackage{algorithm}
\usepackage{algorithmic}
\usepackage{graphicx}
\usepackage{textcomp}
\usepackage{xcolor}
\usepackage{array}
\usepackage{multirow}
\usepackage{subcaption}
\usepackage{booktabs}
\usepackage{cite}
\usepackage[hidelinks]{hyperref}
\newtheorem{proposition}{Proposition}

\begin{document}

\title{Energy-Aware Compression-Computation Co-Adaptation for Latency Minimization in Multi-User Semantic Communication}

\author{Loc X. Nguyen, Yumin Park, Avi Deb Raha, Huy Q. Le, Zhu Han,~\IEEEmembership{Fellow,~IEEE}, \\Eui-Nam Huh,~\IEEEmembership{Senior Member,~IEEE}, and Choong Seon Hong,~\IEEEmembership{Fellow,~IEEE}%

\thanks{Loc X. Nguyen, Yumin Park, Avi Deb Raha, Huy Q. Le, Eui-Nam Huh, and Choong Seon Hong are with the School of Computing, Kyung Hee University, Yongin-si, Gyeonggi-do 17104, Rep. of Korea, e-mail: \{xuanloc088, yumin0906, avi, quanghuy69, johnhuh, cshong\}@khu.ac.kr.}
\thanks{Zhu Han is with the Electrical and Computer Engineering Department, University of Houston, Houston, TX 77004, and also with the Department of Computer Science and Engineering, Kyung Hee University, Yongin-si, Gyeonggi-do 17104, Rep. of Korea, e-mail:{\{hanzhu22\}}@gmail.com}
}

\markboth{Journal of \LaTeX\ Class Files,~Vol.~14, No.~8, August~2021}
{Shell \MakeLowercase{\textit{et al.}}: A Sample Article Using IEEEtran.cls for IEEE Journals}

\maketitle

\begin{abstract}
Deep joint source-channel coding-enabled (DeepJSCC) semantic communication (SemCom) has excelled at delivering high perceptual quality at low channel-bandwidth ratios, which positions it as a pillar for next-generation wireless networks. However, the existing works have difficulty accommodating user heterogeneity in terms of communication channel quality, expected quality-of-service (QoS) targets, and the available local energy. Therefore, in this paper, we explicitly reflect the heterogeneity of user devices in terms of the differences in expected QoS, channel condition, and local energy, and then mathematically formulate the problem. Next, we propose an energy-aware compression-computation co-adaptation (CoCo) framework, in which the base station can meet the expected user QoS by transmitting a longer signal or offloading the task to a local device. The user has to dedicate energy to denoising the signal to recover higher-fidelity latent features before feeding it to the semantic decoder. To solve the formulated problem, we first decompose it into two sub-problems: parameter optimization and resource allocation problems. Specifically, we propose a robust codec that effectively works under a diversity of compression rates and channel noise without re-training, while the greedy sub-carrier allocation lowers the communication time. Finally, we present simulation results on standard image datasets over additive white Gaussian noise to demonstrate the effectiveness of CoCo, which reduces total latency relative to rate-only adaptive DeepJSCC or denoising-only, thereby ensuring the demands of each individual user are met. 
\end{abstract}


\section{Introduction}
\IEEEPARstart{W}{ith} the rapid development of immersive media, machine-to-machine interaction, and edge intelligence, sixth-generation (6G) wireless networks are expected to move beyond the bit-faithful transmission paradigm toward communication that preserves \emph{meaning}~\cite{9770094,9864327,11192484}. Semantic communication (SemCom) can be described as a post-Shannon paradigm in which the transmitter sends a compact signal/relevant representation of the source rather than its exact bit sequence~\cite{qin2021semantic,9955525}. For image and video services in particular, deep joint source-channel coding (DeepJSCC) has emerged as a powerful realization of this idea: a neural encoder maps the source directly to channel symbols and a neural decoder reconstructs it, jointly optimized end-to-end so that perceptual quality remains high under the channel noise rather than collapsing at a cliff~\cite{8723589,11263916,9398576}.

A defining strength of DeepJSCC is its ability to operate at very low channel-bandwidth ratios (CBRs) while retaining high reconstruction quality. Recent designs further make the codec \emph{adaptive}: a single model conditions its behavior on the instantaneous signal-to-noise ratio (SNR)~\cite{9438648} and on a target transmission rate~\cite{9791398,10589474}, and so the same network serves a range of channels and bandwidth budgets without retraining. In a scenario that serves a single user, this rate adaptation can be the solution: the transmitter simply sends as many semantic symbols as the channel and the quality target demand.

However, in practice, a single base station (BS) typically serves many users simultaneously over a shared, finite spectrum. In a realistic cell, users differ in both channel quality (each user experiences its own SNR) and in their quality-of-service (QoS) expectations. Under the rate adaptation scenario~\cite{10382545}, the only way to satisfy a user with a weak channel or a stringent target is to transmit more channel symbols, which consume more communication resources or otherwise increase transmission time~\cite{10327757}. Given the limited set of orthogonal subcarriers and communication resources, the BS's resource allocation must be carefully designed to meet all QoS demands. Regardless of the allocation design, this scenario scales poorly: the delay grows with the number of highly demanding users, the number of users in low-SNR conditions, and is aggravated by the limited available bandwidth.

Therefore, an alternative solution has been studied, inspired by a previous task-offloading approach~\cite{10004947}. Specifically, a user device with increasing computing capacity can utilize computing resources to recover a high-quality signal from a noisy, heavily compressed one. Diffusion and other iterative denoisers have shown that latent representations corrupted by noise can be progressively refined toward a clean manifold at a cost paid in computing resources rather than spectrum alone~\cite{NEURIPS2020_4c5bcfec,Rombach_2022_CVPR,10480348,wang2025latent,11505923}. This reflects a broader shift in SemCom research, from optimizing communication alone toward also exploiting computation. For instance, the work~\cite{10978707} considered the reasoning capacity of the semantic receiver to fill in the missing representations when the communication links are unavailable. While the work~\cite{10999070} utilizes deep reinforcement learning (DRL) to adjust the depth of the transformer model, deeper layers reduce communication but sacrifice computation. Nevertheless, the signal refinement step or deeper blocks consume the device's energy and add processing delay, a cost that neither work accounts for under the heterogeneous energy budgets of a multi-user network.

Consequently, in this paper, we formulate an optimization problem for minimizing the total latency across different users with a wide range of QoS demands. Then, we propose a co-adaptation framework that can effectively allocate communication and computation resources in such a way that all the QoS demands from users are met. Specifically, we propose the iterative latent denoiser running on the user device; a user who can afford the energy may compute more and receive fewer symbols, shortening its communication time and reducing its bandwidth usage. On the other hand, a user with limited energy for denoising can receive a longer signal. Therefore, the BS has to decide the number of local denoising steps, the signal length, and the subcarriers that minimize the total transmission time, subject to each user's QoS and to its device energy budget. The main contributions of this paper are summarized as follows:
\begin{itemize}
  \item We formalize an optimization framework for a semantic communication system that explicitly reflects the user's heterogeneity in terms of the different QoS demands, the wireless channel condition, and the available energy. Specifically, we consider that users with high QoS demands can be facilitated in two ways: either local computing to refine the signal, or a longer received signal that requires more communication resources. To the best of our knowledge, this is the first work to consider dedicating both wireless network resources and user computation resources to address the heterogeneity in user demand.
  \item We design a single DeepJSCC codec with a priority latent representation, conditioned on both SNR and compression rate, so that one model covers all transmission rates and wireless conditions. At the receiver, we also attach the SNR and compression rate to the iterative latent denoiser that performs noise diffusion refinement on the user device. The number of denoising steps is limited by the device's energy, and this problem is a mixed-integer nonlinear program. Therefore, we develop a low-complexity energy-aware greedy allocation that selects each user's rate, steps, and subcarriers to minimize the latency.
  \item Extensive simulations on standard image datasets over an AWGN channel demonstrate that the denoiser improves reconstruction quality most at low rate and low SNR, and the proposed co-adaptation reduces total transmission latency compared with rate-only and denoising-only adaptive DeepJSCC scenarios, while ensuring the QoS for every single user in the limited wireless bandwidth.
\end{itemize}

The rest of this paper is organized as follows. Section~\ref{sec:related} reviews related work. Section~\ref{sec:sysmodel} presents the multi-user system model, including the adaptive-rate codec, the shared-subcarrier transmission model, the on-device iterative-denoising receiver, and the latency model. Section~\ref{sec:problem} formalizes the compression-computation complementation and the latency-minimization problem. Section~\ref{sec:method} details the proposed CoCo framework, its two-stage training, and the greedy allocation. Section~\ref{sec:results} reports results, and finally Section~\ref{sec:conclusion} concludes the paper.

\section{Related Work}\label{sec:related}
\subsection{Deep Joint Source-Channel Coding}
The term DeepJSCC was introduced by Bourtsoulatze~\emph{et al.}~\cite{8723589}, who showed that a convolutional encoder-decoder trained end-to-end over a noisy channel outperforms separate source and channel coding for wireless image transmission, particularly in the low-SNR and low-bandwidth regimes and without a cliff effect. Various works have inherited the idea of joint source-channel optimality and developed further techniques to improve the performance, such as work~\cite{9746335}, which considers channel feedback in the encoding/decoding process to improve the system's robustness to noise. In addition, the authors of~\cite{10094735} adopted the Swin-transformer model for the source encoder/decoder to better capture the semantic meaning, while having lower computing complexity compared to conventional transformers. These aforementioned works mainly focus on image transmission, and some works has extended it to text~\cite{9398576,9791409} and speech~\cite{9450827}. However, these works optimize the encoder-decoder pair and implicitly treat the receiver as a single deterministic forward pass, overlooking user heterogeneity in QoS demand, local computations, and their potential to improve system performance.

\subsection{Rate and SNR Adaptive SemCom}
To avoid storing one set of parameters for the model to deal with various channel conditions, adaptive DeepJSCC conditions a single network on side information. To be specific, Attention DeepJSCC (ADJSCC)~\cite{9438648} reweights feature maps by the instantaneous SNR through squeeze-and-excitation-style modules, so one model can adjust its encoding/decoding process to adapt to a wide range of SNRs. Rate adaptation has been achieved by content-aware symbol allocation~\cite{9791398}, by transformer models that expose multiple rates~\cite{10589474}, by adaptive rate control~\cite{10436878}, and by progressive or successive-refinement transmission~\cite{8815416,9464731}. Specifically, nonlinear transform source-channel coding~\cite{9791398} introduced learned entropy models that allocate channel symbols across the latent according to its content, yielding state-of-the-art rate-distortion performance. Adopting the advancement in DL architecture,~\cite{10436878} leveraged the Swin transformer to propose a novel bandwidth and channel-quality-adaptive scheme.~\cite{8815416} presented the first work that considers progressive image transmission with different complexities; later in~\cite{9464731}, they developed a more general system scenario, where the images are transmitted progressively in layers, either in order or in any order. Recently, the author in~\cite{11587127} proposed DeepJSCC for satellite communication, which adjusts the compression rate to meet the sensing requirement for satellite tasks. Closest to our proposal, the predictive and adaptive deep coding (PADC) framework~\cite{10015684} selects, for a single image, the minimal code rate that meets a target peak signal-to-noise ratio (PSNR), thereby minimizing bandwidth under a quality constraint. These designs enable the transmitter to adjust the amount of transmitted information to achieve a target reconstruction quality. However, in all of these approaches, the desired quality is improved solely by allocating more channel symbols. When multiple users compete for limited spectrum resources, improving the performance of one user inevitably requires reducing the transmission resources available to others. In contrast, the complementary option of exploiting receiver-side computation has not been explored.

\subsection{Generative and Diffusion-Aided Receivers}
Denoising diffusion probabilistic models (DDPMs)~\cite{NEURIPS2020_4c5bcfec} and latent diffusion models~\cite{Rombach_2022_CVPR} generate or restore signals by iteratively removing noise, with a number of reverse steps that directly trade computation for output quality. This iterative, compute-scalable structure has begun to enter wireless communications: channel denoising diffusion models (CDDM)~\cite{10480348} learn to remove channel-induced noise from received symbols, latent-diffusion receivers perform channel-adaptive equalization and denoising~\cite{10910031}, and generative SemCom~\cite{11505923,10158995} reconstructs perceptually faithful content from minimal transmitted information. Similarly, authors in~\cite{11578137} designed a one-step diffusion model at the receiver to provide a reliable semantic communication system, while~\cite{10891405} considers the number of denoising steps conditioned on a similarity score between the received and transmitted features. Despite this progress, diffusion-aided receivers have been studied as quality enhancers for a single compression rate; their per-user step count has been overlooked and not considered together with the rate-adaptive approach to facilitate heterogeneous user demand. In this paper, we formulate a problem that combines both approaches, which can complement each other in bandwidth-limited or energy-limited scenarios.

\subsection{Resource Allocation for Multi-User SemCom}
A growing literature allocates physical-layer resources for SemCom, e.g., power, subcarriers, or the number of transmitted semantic symbols, to maximize a semantic-aware QoS or quality of experience (QoE)~\cite{9763856,9955525}, and serves multiple users with task-oriented semantic codecs~\cite{9830752}. Latency- and energy-aware formulations are likewise central to SemCom-enabled networks, where minimizing transmission delay under fidelity constraints is the design goal~\cite{11039171}. These formulations optimize how transmission resources are split among users.~\cite{11274758} proposed a DRL framework to select the compression rate for each individual user to balance semantic accuracy, latency, and energy consumption. A few recent works do consider computation alongside communication: receiver reasoning can recover undelivered information in multi-user SemCom~\cite{10978707}, and a computation-communication tradeoff metric can be optimized via deep reinforcement learning~\cite{10999070}. However, these either assume a fixed receiver pipeline and optimize a single transmission axis, or trade computation for communication through task-level reasoning or a system-level resource metric rather than through a runtime, retransmission-free refinement knob. None couples the rate-adaptive image codec with an on-device iterative denoiser and co-allocates the denoiser's per-user step budget-bounded by device energy-against transmission time under a shared spectrum so as to minimize total latency, which is precisely the joint problem this paper addresses. We bridge the adaptive-codec, diffusion-receiver, and resource-allocation strands by exposing on-device computation as a substitute for communication time and allocating the two jointly.

\section{System Model}\label{sec:sysmodel}
\subsection{Multi-User Scenario}
We consider the downlink of a single BS serving a set $\mathcal{K}=\{1,2,\dots,K\}$ of users over a shared set $\mathcal{N}=\{1,2,\dots,N\}$ of orthogonal subcarriers, as illustrated in Fig.~\ref{fig:scenario}. As in orthogonal frequency-division multiple access (OFDMA), each subcarrier of bandwidth $B$ is assigned to at most one user at a time, and so the links do not interfere and the spectrum the BS spends is the shared, finite pool of $N$ subcarriers. The BS has a source image $\mathbf{s}_k \in \mathbb{R}^{C_{\mathrm s}\times H\times W}$ and wants to transmit to user $k$, where $C_{\mathrm s}=3$ is the number of color channels and $H,W$ are the spatial dimensions. Each user $k$'s heterogeneity is characterized by three factors: its channel SNR $\rho_k$ assumed known due to channel estimation/feedback; its quality target expressed as a minimum QoS $QoS_k^{\min}$ on a single scalar quality metric; and its device energy budget $E_k$. In general, the heterogeneity of users is reflected in three aspects: wireless channel condition, expected QoS, and available device energy. The BS employs a single semantic-channel encoder $f_\theta$, while the semantic-channel decoder $f_\psi$ and an iterative denoiser $g_\phi$ are deployed at the users. For each individual user, the BS has to select three control variables: a transmission rate $r_k$, the number of receiver denoising steps $m_k$, and finally the subcarriers allocated to that user so that the user's expected quality is guaranteed. 

\begin{figure*}[t]
\centering
    \includegraphics[width=1\textwidth]{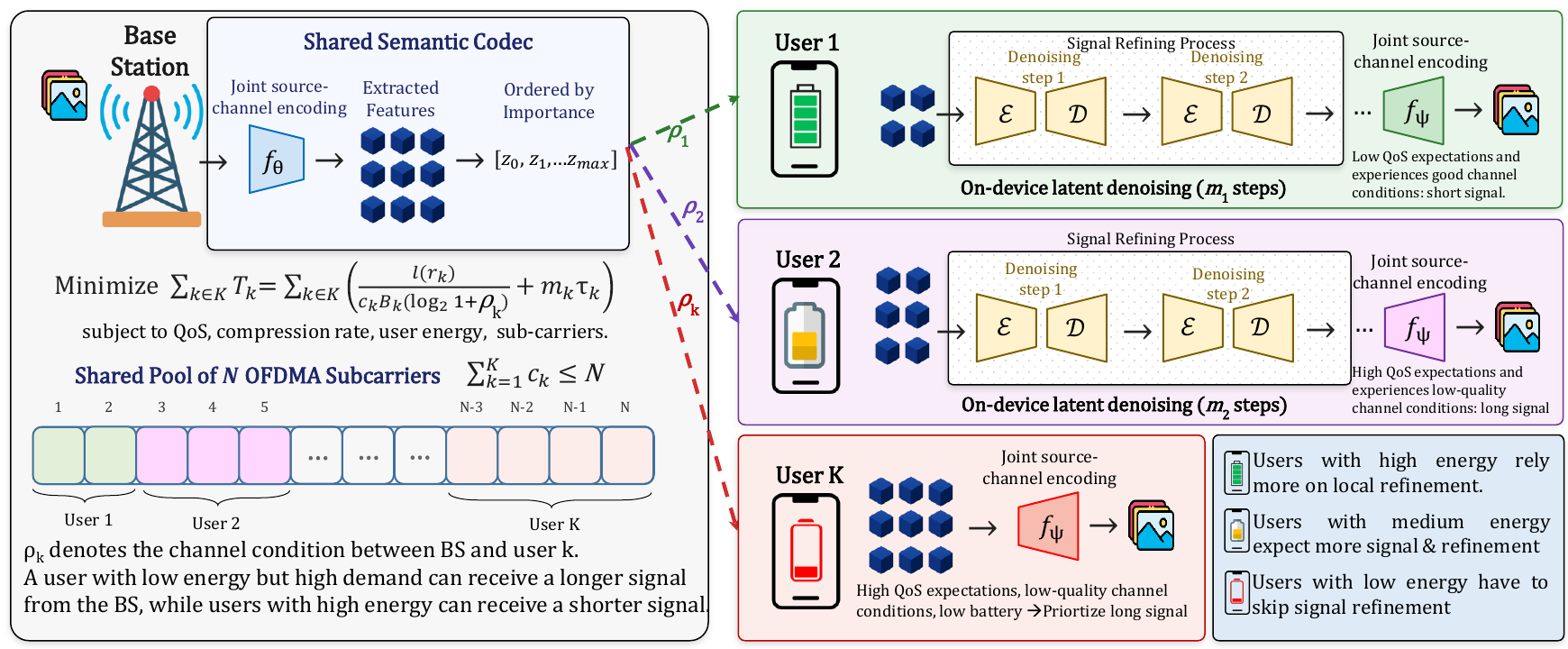}
\caption{Our proposed system: each user can reach its quality target by receiving more transmitted symbols or by performing more on-device denoising. In a multiple-user setting, the BS has to allocate the communication resource among users and offload the computation task to local devices. }
\label{fig:scenario}
\end{figure*}

\subsection{Adaptive-Rate Semantic Encoding}
The BS encoder maps the source image to a latent space:
\begin{equation}
\mathbf{Z}_k = f_\theta(\mathbf{s}_k;\,\rho_k,\,r_k) \in \mathbb{R}^{C_{\max}\times H'\times W'},
\label{eq:encoder}
\end{equation}
where $C_{\max}$ is the maximum number of latent channels, $H'=H/2^{J}$ and $W'=W/2^{J}$ are the spatial dimensions after downsampling stages, and $\theta$ denotes the encoder parameters. The encoder is conditioned on both the SNR $\rho_k$ and the target rate $r_k$ through feature-modulation modules, so that a single network can work across channels and rates.

After the semantic features are extracted, they are ordered by their importance, so transmitting only the first $r_k$ of them produces a valid lower-rate codeword. The transmitted latent is the prefix:
\begin{equation}
\mathbf{Z}_k^{(r_k)} = \mathbf{Z}_k[1{:}r_k,:,:],\qquad r_k \in \mathcal{R},
\label{eq:slice}
\end{equation}
where $\mathcal{R}=\{r^{(1)}<\dots<r^{(L)}\}\subseteq\{1,\dots,C_{\max}\}$ is the discrete set of supported rates. After power normalization to unit average power, the prefix is mapped to a complex channel-input vector $\mathbf{x}_k\in\mathbb{C}^{n_k}$ of length $n_k$. The CBR is denoted as:
\begin{equation}
\mathrm{CBR}_k \;=\; \frac{n_k}{C_{\mathrm s} H W} \;=\; \frac{r_k\,H'W'}{2\,C_{\mathrm s} H W}\;\propto\; r_k,
\label{eq:cbr}
\end{equation}
linear in the rate $r_k$. The number of symbols the BS must transmit over the wireless environment for user $k$ is therefore controlled directly by how many latent channels are transmitted. The transmitted payload size $\ell(r_k)=\mu\, n_k \propto r_k$, where $\mu$ is the number of bits per channel symbol after digital mapping.

\subsection{Shared-Subcarrier Transmission Model}
Following the OFDMA model, let $a_{k,n}\in\{0,1\}$ indicate that subcarrier $n$ is assigned to user $k$, and let $c_k=\sum_{n\in\mathcal{N}}a_{k,n}$ be the number of subcarriers allocated to user $k$. Each subcarrier serves at most one user, as expressed below:
\begin{equation}
\sum_{k\in\mathcal{K}} a_{k,n} \le 1,\quad \forall n\in\mathcal{N},
\qquad
\sum_{k\in\mathcal{K}} c_k \le N,
\label{eq:subcarrier}
\end{equation}
where the second inequality indicates that the number of allocated subcarriers is less than or equal to the shared spectrum. With $c_k$ subcarriers and the SNR $\rho_k$ value, the achievable downlink rate from BS to user $k$ is calculated as~\cite{goldsmith2005overview}:
\begin{equation}
R_k \;=\; c_k\, B \,\log_2\!\big(1+\rho_k\big),
\label{eq:rate}
\end{equation}
in bits/s, consistent with the latency model of prior SemCom resource-allocation studies~\cite{11039171}. As a result, the transmission time from BS to user $k$ is given as:
\begin{equation}
t_k^{\mathrm{tx}} \;=\; \frac{\ell(r_k)}{R_k} \;=\; \frac{\mu\, r_k H'W'/2}{c_k\, B_k \, \log_2(1+\rho_k)},
\label{eq:ttx}
\end{equation}
which depends on the number of transmitted latent features and the subcarrier bandwidth. With the estimated SNR $\rho_k$ for the wireless channel quality, the BS can approximate the achievable rate, predict reconstruction quality at the user under those particular conditions, and determine the lowest CBR so the user meets the QoS with the minimum latency.

\subsection{Receiver: On-Device Iterative Latent Denoising}
The received latent features are corrupted by two effects: channel fading and noise. Depending on the BS decision, the number of transmitted prefixes over a channel is denoted as 
\begin{equation}
\mathbf{y}_k = h_k\,\mathbf{x}_k + \mathbf{n}_k,
\label{eq:channel}
\end{equation}
where $h_k$ is the channel fading (it is equal to 1 for AWGN and follows $\mathcal{CN}(0,1)$ for Rayleigh fading) and $\mathbf{n}_k \sim \mathcal{CN}(0,\sigma_k^2\mathbf{I})$ with $\sigma_k^2 = 10^{-\rho_k/10}$ under unit-power transmission. The received symbols are reshaped into a noisy latent $\widehat{\mathbf{Z}}_k^{(r_k)}\in\mathbb{R}^{r_k\times H'\times W'}$ and zero-padded along the channel axis to the full width $C_{\max}$, giving a fixed-size tensor $\widehat{\mathbf{Z}}_k^{(0)}\in\mathbb{R}^{C_{\max}\times H'\times W'}$. 

On the local device side, user $k$ can mitigate the channel noise with an iterative latent denoiser $g_\phi$. This module is conditioned on the SNR and compression rate to facilitate the ability to progressively refine the received features in a wide range of conditions. Starting from $\widehat{\mathbf{Z}}_k^{(0)}$, it applies $m_k$ residual refinement steps as in the following equation:
\begin{equation}
\widehat{\mathbf{Z}}_k^{(j)} = \widehat{\mathbf{Z}}_k^{(j-1)} + g_\phi\!\left(\widehat{\mathbf{Z}}_k^{(j-1)};\,\rho_k,\,r_k\right),\quad j=1,\dots,m_k,
\label{eq:denoise}
\end{equation}
where $\phi$ denotes the denoiser parameters and $m_k\in\mathcal{M}=\{0,1,\dots,M\}$ is the number of denoising steps. Each step pushes the latent toward the clean feature manifold, in the spirit of the progressive denoising performed by diffusion models~\cite{NEURIPS2020_4c5bcfec,Rombach_2022_CVPR}. The refined latent is then decoded:
\begin{equation}
\hat{\mathbf{s}}_k = f_\psi\!\left(\widehat{\mathbf{Z}}_k^{(m_k)};\,\rho_k,\,r_k\right),
\label{eq:decoder}
\end{equation}
where $\psi$ denotes the decoder parameters and $\hat{\mathbf{s}}_k$ is the reconstructed image.

\subsection{Quality, Latency, and Energy}
With the semantic communication framework in place, we define a training loss to optimize the parameters of the encoder, decoder, and denoising module as follows: 
\begin{equation}
    \mathcal{L}= \textrm{MSE} (s_{k}, \hat{s}_{k}),\label{eq:MSE}
\end{equation}
where the objective is that the reconstructed image is as close to the original as possible. Here, we denote the QoS by a specified single scalar quality metric as  $D(\mathbf{s}_k,\hat{\mathbf{s}}_k)$. Since the channel fading and noise of the wireless environment are random, the QoS for the user can be written as the expected function as follows:
\begin{equation}
\textrm{QoS}(r,m;\rho) \;=\; \mathbb{E}_{\mathbf{s},\,\mathbf{n}}\!\left[\, D\big(\mathbf{s},\,\hat{\mathbf{s}}(r,m;\rho)\big)\,\right],
\label{eq:Dbar}
\end{equation}
i.e., the average QoS when rate $r$ is transmitted and $m$ denoising steps are applied at SNR $\rho$. Sending more symbols can provide extra information for the decoder to interpret, so $\textrm{QoS}(\cdot,\cdot;\rho)$ is increasing in $r$. Similarly, the step count $m$ can provide extra signal refinement, and we further enforce it by the training design so that each additional denoising step is expected to increase the QoS.

Each denoising step in Eq.~\eqref{eq:denoise} costs a fixed amount of computation $\phi_0$, executed on user $k$'s device, which processes $F_k$ operations per second at an energy cost of $\kappa_k$ per operation. The time and energy consumption for one step are given as:
\begin{equation}
\tau_k = \frac{\phi_0}{F_k},\qquad e_k = \kappa_k\,\phi_0,
\label{eq:taue}
\end{equation}
and so $m_k$ steps incur a processing delay $m_k\tau_k$ and consume energy $m_k e_k$. The number of denoising steps at the local device is bounded by the available energy budget:
\begin{equation}
m_k\, e_k \le E_k \;\Longleftrightarrow\; m_k \le M_k \triangleq \big\lfloor E_k/e_k \big\rfloor,
\label{eq:energy}
\end{equation}
where $M_k$ is the maximum number of steps user $k$ can afford. The \emph{end-to-end delivery latency} of user $k$ is the sum of its communication time and its on-device processing delay,
\begin{equation}
T_k \;=\; t_k^{\mathrm{tx}} + m_k\tau_k \;=\; \frac{\ell(r_k)}{c_k\, B\, \log_2(1+\rho_k)} + m_k\tau_k,
\label{eq:Tk}
\end{equation}
which reduces to the pure transmission time when $m_k=0$. Therefore, the total latency among all the users in the network can be calculated as:
\begin{equation}
T(\{r_k,m_k,c_k\}) \;=\; \sum_{k\in\mathcal{K}} T_k.
\label{eq:Ttot}
\end{equation}


\section{The Compression-Computation Co-adaptation and Problem Formulation}\label{sec:problem}

\subsection{Per-User Operating Points}
For a user with channel condition $\rho$, a target QoS requirement $\Psi$, and an available energy budget $E$, the BS must determine how to satisfy the required reconstruction quality. When transmitting the default number of semantic symbols is insufficient to achieve the target QoS, the system has two options. The first way is to increase the transmission length by sending more latent symbols so that the receiver can interpret the feature correctly, which improves reconstruction quality at the cost of additional communication resources. The second way is to keep the transmission rate low and offload the task to the user device, where it performs iterative latent denoising to remove the noise from the wireless environment before decoding, thereby improving the QoS through additional local computation. The first option consumes more wireless resources, whereas the second consumes more computation time and energy on the user device. Based on these two mechanisms, we define the feasible operating set:
\begin{equation}
\begin{aligned}
\mathcal{F}(\rho,\Psi,E)
= \big\{(r,m)\in\mathcal{R}\times\mathcal{M}: \textrm{QoS}\ge \Psi, m\le M\big\},
\end{aligned}
\label{eq:feasible}
\end{equation}
which contains all compression rate and denoising-step pairs that satisfy the user's QoS requirement while remaining within the available energy budget.

Instead of relying on a single approach, we adopt a hybrid approach that addresses the user demand by using both the compression rate and local denoising computation. Specifically, our feasible set provides a comprehensive relation between the compression rate and the number of denoising steps so that the BS can choose to balance communication and computation for the user to minimize the latency. Specifically, for the user's QoS, we need to determine the optimal pair of compression rate and denoising steps so that we do not overconsume the bandwidth resource or overuse the local computation. For example, when the actual QoS of the user is just a little behind the demand, and the communication resource is not available for extra signal, we rely on the denoising module. On the other hand, when the demand cannot be met by using a single approach, we can utilize both. Finally, for devices with limited energy, the BS can actively transmit at higher rates.

\subsection{Latency Minimization Problem}
With the knowledge of expected QoS, the available energy, and the wireless channel condition for each individual user, the BS has to decide the compression rate $\{r_{k}\}$ and the number of denoising steps $\{m_{k}\}$ at local devices, and a subcarrier assignment $\{a_{k,n}\}$, so that it can accommodate all the user demand for image quality while minimizing the to the average total latency across user devices. In addition, we have to obtain the optimal parameter sets for the BS encoder, the plug-and-play denoising module, and finally the user decoder. Therefore, we can formulate the problem as:
\begin{equation}
\begin{aligned}
\min_{\theta, \phi, \psi, \mathbf{r},\mathbf{m},\mathbf{a}}\;\; & \sum_{k\in\mathcal{K}} T_k 
   = \sum_{k\in\mathcal{K}}\!\left(\frac{\ell(r_k)}{c_k B_k \log_2(1+\rho_k)} + m_k\tau_k\right) \\[-1pt]
\text{s.t.}\;\; & C_1:\; \textrm{QoS}(r_k,m_k;\rho_k) \geq \Psi_k,  \forall k\in\mathcal{K},\\
& C_2:\; m_k\, e_k \le E_k, \forall k\in\mathcal{K},\\
& C_3:\; \sum_{k\in\mathcal{K}} a_{k,n} \le 1,  \forall n\in\mathcal{N},\\
& C_4:\; c_k=\!\!\sum_{n\in\mathcal{N}}\! a_{k,n} \ge 1, \forall k\in\mathcal{K},\\
& C_5:\; r_k\in\mathcal{R},\; m_k\in\mathcal{M},\; a_{k,n}\in\{0,1\}, \forall k,n,
\end{aligned}
\tag{P1}
\label{eq:P1}
\end{equation}
where $C_1$ enforces each user's QoS, $C_2$ caps on-device computation by the device energy budget, $C_3$ keeps the subcarrier allocation within the shared spectrum, $C_4$ guarantees each served user at least one subcarrier, and $C_5$ restricts the controls to the supported discrete sets. Problem (P1) makes the compression-computation tradeoff operational: a user may shorten $T_k$ by lowering $r_k$ (fewer symbols), but only if it spends enough steps $m_k$-paid in energy through $C_2$ and in processing delay through the $m_k\tau_k$ term-to keep $C_1$ satisfied, or else by acquiring more subcarriers $c_k$ from the shared, capacity-limited pool through $C_3$. 

Problem (P1) is a mixed-integer nonlinear program (MINLP), the binary subcarrier-assignment variables $a_{k,n}$ with discrete variables $(r_k,m_k)$, and its objective is nonlinear in these variables. Three features make it intractable to solve directly. First, the $\textrm{QoS}(r,m;\rho)$ has no closed form; it is defined by the learned codec-denoiser pair, so the QoS region $C_1$ must be obtained empirically. Second, the subcarrier assignment under $C_3$-$C_4$ is itself a combinatorial allocation, and jointly choosing assignments and per-user operating points is NP-hard in general~\cite{martello1990knapsack}. Third, and unlike the bandwidth-only setting, the problem \emph{does not cleanly decouple}: because the communication term scales as $1/c_k$, the latency for a user depends on how many subcarriers it receives.

\section{Proposed Framework: CoCo}\label{sec:method}

The proposed compression-computation co-adaptation framework comprises three parts: an adaptive-rate DeepJSCC that effectively adjusts transmission length is conditioned on SNR and rate; the iterative latent denoiser that runs on the user device with the number of denoising steps bounded by the device's energy; and finally, an allocation module that assigns each user a compression rate, steps, and a number of subcarriers by solving~(\ref{eq:P1}). We divide the problem into two sub-problems: the first one is the parameter optimization for ($\theta,\psi,\phi$); the second sub-problem is resource allocation for users in the network by controlling the compression rate, number of denoising steps, and the sub-carriers ($\mathbf{r},\mathbf{m},\mathbf{a}$). The first sub-problem is trained offline and then uses the frozen parameters for the scheduling problem, which is an online allocation phase that can be repeated over time.

\subsection{Adaptive Coder, Denoising Model, and Training}

\subsubsection{Adaptive-Rate Codec with Priority}
The encoder and decoder are convolutional networks conditioned on the SNR $\rho$ and the compression rate $r$, which enable a single model to serve all channel conditions and communication rates. Specifically, the extracted features are ordered by importance, and then, based on the determined compression rate, the model transmits only the number of features that correspond to the rate. During training, we randomly sample the rate for each batch to order important content in sequence. Therefore, at the inference stage, changing the rate simply changes the number of features being transmitted, which eliminates the need for retraining.

\subsubsection{On-Device Iterative Latent Denoiser}
We construct a small denoiser network to refine the received features based on the latent denoising model~\cite{Rombach_2022_CVPR,10480348,NEURIPS2020_4c5bcfec}. In general, our diffusion only has a total of four layers: one input projection layer, two convolution layers, and finally an output projection with a skip connection, which makes our denoising model lightweight and suitable for a user device. Additionally, we conditioned the denoiser with two pieces of information: noise level and communication rate. Each piece of conditioning information is useful in a different way for the denoising model: the noise level signals how corrupted the signal may be during the wireless transmission; the communication rate informs the network which units to refine and which units are zero-padded rather than transmitted features. It is worth noticing that the same weight of the denoising model are applied repeatedly to the received feature to remove noise for multiple steps, which is a unique property of the diffusion model. Ideally, more steps can yield better quality at the cost of more computing and energy consumed, which is constrained by the user's energy budget.

\subsubsection{Two-Stage Training}\label{sec:training}
We train CoCo in two stages, summarized in Algorithm~\ref{alg:train}. In Stage~1, the encoder and decoder are trained end-to-end over the channel with the denoiser bypassed ($m=0$); each mini-batch samples a random rate $r\in\mathcal{R}$ and a random SNR $\rho$ from the operating range, so the single codec learns to serve all rates and channels. The objective is the expected loss
\begin{equation}
\mathcal{L}_{\mathrm{codec}}(\theta,\psi) = \mathbb{E}_{\mathbf{s},\,r,\,\rho,\,\mathbf{n}}\big[\, \ell_{\mathrm{tr}}\big(\mathbf{s},\,\hat{\mathbf{s}}(r,0;\rho)\big)\,\big],
\label{eq:loss_codec}
\end{equation}
where $\ell_{\mathrm{tr}}$ is conventional training loss as shown in Eq.~(\ref{eq:MSE}). After Stage~1, we obtain a complete and self-contained adaptive DeepJSCC system. In Stage~2, the encoder $\theta$ is frozen, while the parameters of the denoiser $\phi$ and the decoder $\psi$ are jointly optimized over a wide range of channel conditions and compression rates. As described above, the number of denoising steps is determined by BS and the user's energy rather than a fixed value; therefore, the optimization loss is given as follows:
\begin{equation}
\mathcal{L}_{\mathrm{rx}}(\phi,\psi) = \mathbb{E}_{\mathbf{s},\,r,\,\rho,\,\mathbf{n},\,m_{\mathrm{tr}}}\big[\, \ell_{\mathrm{tr}}\big(\mathbf{s},\,\hat{\mathbf{s}}(r,m_{\mathrm{tr}};\rho)\big)\,\big],
\label{eq:loss_den}
\end{equation}
where $m_{\mathrm{tr}}\!\sim\!\mathcal{U}\{0,\dots,M\}$. The above equation implies that optimization of the network is over all the possible values for the number of denoising steps, rather than the fixed one.

\begin{algorithm}[t]
\caption{Two-Stage Training of CoCo}
\label{alg:train}
\begin{algorithmic}[1]
\REQUIRE Training images $\mathcal{D}$; rate set $\mathcal{R}$; SNR range; maximum step count $M$
\ENSURE Encoder $\theta^\star$; receiver module $(\phi^\star,\psi^\star)$
\STATE \textbf{Stage 1: Codec (denoiser bypassed)}
\STATE Initialize $\theta,\psi$
\FOR{each mini-batch $\mathbf{s}\sim\mathcal{D}$}
  \STATE Sample rate $r\sim\mathcal{U}(\mathcal{R})$ and SNR $\rho$ from set of values
  \STATE Encode image follow~Eq.~\eqref{eq:encoder}, slice to $r$ as in Eq.~\eqref{eq:slice}, pass through wireless channel, decode with Eq.~\eqref{eq:decoder}.
  \STATE Update $\theta,\psi$ by minimizing $\mathcal{L}_{\mathrm{codec}}$ from Eq.~\eqref{eq:loss_codec}.
\ENDFOR
\STATE \textbf{Stage 2: Receiver refinement}
\STATE Freeze $\theta\leftarrow\theta^\star$; keep $\psi$ from Stage 1; initialize $\phi$
\FOR{each mini-batch $\mathbf{s}\sim\mathcal{D}$}
  \STATE Sample rate $r\sim\mathcal{U}(\mathcal{R})$, SNR $\rho$, and step count $m_{\mathrm{tr}}\sim\mathcal{U}\{0,1,\dots,M\}$
  \STATE Encode, slice, pass channel; unroll denoiser as in Eq.~\eqref{eq:denoise} for $m_{\mathrm{tr}}$ steps; decode
  \STATE Update $\phi$ \emph{and} $\psi$ by minimizing $\mathcal{L}_{\mathrm{rx}}$ from Eq.~\eqref{eq:loss_den}
\ENDFOR
\RETURN $\theta^\star,\phi^\star,\psi^\star$
\end{algorithmic}
\end{algorithm}

\subsection{Offline Operating-Point Profiling}
Since $\textrm{QoS}(r,m;\rho)$ has no closed form, we have to characterize it empirically in offline mode, after the parameter training. Specifically, for each compression rate $r\in\mathcal{R}$, each step count $m\in\mathcal{M}$, and each SNR $\rho$ in $\mathcal{S}$, we evaluate the performance of the trained encoder-denoiser-decoder chain over a set of images. Then, we  average the quality of all the reconstructed images and record it into a lookup table:
\begin{equation}
Q[r,m,\rho] \;\triangleq\; \widehat{\mathbb{E}}_{\mathbf{s},\mathbf{n}}\big[\,D\big(\mathbf{s},\hat{\mathbf{s}}(r,m;\rho)\big)\,\big],\;(r,m,\rho)\in\mathcal{R}\times\mathcal{M}\times\mathcal{S},
\label{eq:Qtable}
\end{equation}
where $\widehat{\mathbb{E}}$ denotes the empirical average. This table captures how image quality depends on compression rate, channel condition, and denoising step count.
\begin{algorithm}[t]
\caption{Energy-Aware Greedy Latency Allocation}
\label{alg:alloc}
\begin{algorithmic}[1]
\REQUIRE Users $\mathcal{K}$ with $(\rho_k,\Psi_k,E_k)$, profiled table $Q[r,m,\rho]$, $\mathcal{R},\mathcal{M}$, subcarriers $N$, bandwidth $B$, per-step $(\tau_k,e_k)$
\ENSURE Operating points $\{(r_k,m_k)\}$ and subcarrier counts $\{c_k\}$
\STATE \textbf{Step 1: Feasible sets}
\FOR{each user $k\in\mathcal{K}$}
  \STATE $M_k \leftarrow \lfloor E_k/e_k\rfloor$
  \STATE $\mathcal{F}_k \leftarrow \{(r,m)\in\mathcal{R}\times\mathcal{M} : Q[r,m,\rho_k]\ge \Psi_k,\, m\le M_k\}$
\ENDFOR
\STATE $\mathcal{K}_{\mathrm s} \leftarrow$ admitted users; \textbf{assert} $N \ge |\mathcal{K}_{\mathrm s}|$
\STATE \textbf{Step 2-3: Initialize and greedily allocate}
\FOR{each $k\in\mathcal{K}_{\mathrm s}$}
  \STATE $c_k \leftarrow 1$; compute $T_k^\star(c_k)$, $T_k^\star(c_k{+}1)$ via~\eqref{eq:bestresp}; $\Delta_k \leftarrow T_k^\star(c_k)-T_k^\star(c_k{+}1)$
\ENDFOR
\STATE Build max-priority queue $\mathcal{Q}$ over $\Delta_k$; $N_{\mathrm{rem}} \leftarrow N -|\mathcal{K}_{\mathrm s}|$
\WHILE{$N_{\mathrm{rem}} > 0$ \AND $\max_k \Delta_k > 0$}
  \STATE $k^\star \leftarrow \arg\max_k \Delta_k$ (pop from $\mathcal{Q}$)
  \STATE $c_{k^\star} \leftarrow c_{k^\star}+1$; \; $N_{\mathrm{rem}} \leftarrow N_{\mathrm{rem}}-1$
  \STATE recompute $T_{k^\star}^\star(c_{k^\star}{+}1)$; \; $\Delta_{k^\star} \leftarrow T_{k^\star}^\star(c_{k^\star})-T_{k^\star}^\star(c_{k^\star}{+}1)$; push to $\mathcal{Q}$
\ENDWHILE
\STATE Set $(r_k,m_k)\leftarrow(r_k(c_k),m_k(c_k))$ for all $k\in\mathcal{K}_{\mathrm s}$
\RETURN $\{(r_k,m_k)\}$, $\{c_k\}$
\end{algorithmic}
\end{algorithm}

\subsection{Energy-Aware Greedy Latency Allocation}\label{sec:alloc}
Substituting the trained parameters and the profiled table into (P1), we obtain the online resource-allocation problem:
\begin{equation}
\begin{aligned}
\min_{\mathbf{r},\mathbf{m},\mathbf{a}}\;\; & \sum_{k\in\mathcal{K}}\!\left(\frac{\ell(r_k)}{c_k B \log_2(1+\rho_k)} + m_k\tau_k\right) \\[-1pt]
\text{s.t.}\;\; & C_1'\!:\; Q[r_k,m_k,\rho_k] \geq \Psi_k,\;\; \forall k\in\mathcal{K},\\
& C_2\text{-}C_5 \text{ of (P1)},
\end{aligned}
\tag{P2}
\label{eq:P2}
\end{equation}
in which the intractable learned constraint $C_1$ of (P1) is replaced by the exact table lookup $C_1'$. Therefore, the decomposition can be considered lossless with respect to the QoS constraint. With the profiled table in hand, the BS solves (P2) with the low-complexity greedy algorithm summarized in Algorithm~\ref{alg:alloc}. It proceeds in three steps.

\emph{Step 1 (Feasible Sets):} For each user $k$, we construct its feasible operating set $\mathcal{F}_k$ by intersecting the QoS region $\{(r,m):\bar D[r,m,\rho_k]\le D_k^{\max}\}$ with the available energy $m\le M_k$. Let $\mathcal{K}_{\mathrm s}$ denote the admitted users and assume $N\ge|\mathcal{K}_{\mathrm s}|$ so each user can be allocated at least one sub-carrier.

\emph{Step 2 (Per-user best response):} For a given number of subcarriers $c$, we determine the rate-step pair $(r,m)$ that minimizes user $k$'s latency, along with the resulting latency value:
\begin{equation}
T_k^\star(c) = \!\!\min_{(r,m)\in\mathcal{F}_k}\! \left[\frac{\ell(r)}{c\,B\,\varepsilon_k} + m\,\tau_k\right],\quad \varepsilon_k\triangleq\log_2(1+\rho_k),
\label{eq:bestresp}
\end{equation}
where $(r_k(c),m_k(c))$ is the minimizing pair. As $c$ grows, the communication time reduces, so a user given more sub-carriers can receive more signal and denoise less. This dependence is exactly the coupling that prevents (P1) from decoupling. Since $\mathcal{F}_k$ is small, $T_k^\star(c)$ and its minimizer are obtained by a direct scan.

\emph{Step 3 (Greedy subcarrier allocation):} The BS first gives every admitted user one subcarrier, guaranteeing the constraint $C_4$, and then distributes the remaining  $N-|\mathcal{K}_{\mathrm s}|$ subcarriers one at a time. At each step, it computes the latency reduction for every user due to the additional subcarrier:
\begin{equation}
\Delta_k(c_k) = T_k^\star(c_k) - T_k^\star(c_k+1) \ge 0,
\label{eq:marginal}
\end{equation}
and then allocate the next subcarrier to the user with the largest reduction, then update the user's $c_k$, operating point, and marginal value. Using a max-priority queue keyed by $\Delta_k$, each assignment costs one pop and one push. The algorithm loop stops when all the sub-carriers have been allocated.

\subsubsection{Computational Complexity}
The resource allocator runs in two phases, which we provide the cost separately.

\emph{Feasible Set:} For every user the BS forms the feasible set $\mathcal{F}_k$ and reduces it to its Pareto-optimal latency envelope, the operating points $(r,m)$ that are not dominated in both communication and denoising delay. Each user costs $\mathcal{O}(|\mathcal{R}|\,|\mathcal{M}|)$, so the setup over all $K$ users costs $\mathcal{O}(K\,|\mathcal{R}|\,|\mathcal{M}|)$.

\emph{Greedy loop:} The loop distributes at most $N$ subcarriers, one per iteration. Each iteration extracts the user of largest marginal gain from the max-priority queue and updates its entry, which takes one queue pop/push at $\mathcal{O}(\log K)$ and one re-scan of that user's pruned envelope at $\mathcal{O}(|\mathcal{R}|)$ to recompute its best response. The loop therefore costs $\mathcal{O}\!\big(N(|\mathcal{R}|+\log K)\big)$. Treating the codec-set sizes $|\mathcal{R}|$ and $|\mathcal{M}|$ as small constants, the total per-decision cost is given as follows:
\begin{equation}
\mathcal{O}\!\big(K\,|\mathcal{R}|\,|\mathcal{M}| + N\log K\big),
\label{eq:complexity}
\end{equation}
which is linear in the number of users $K$ and in the number of subcarriers $N$. For comparison, exhaustive search conducts every combination of per-user operating points, costing $\mathcal{O}\!\big(\prod_k|\mathcal{F}_k|\big)$, which grows exponentially in $K$. On the other hand, a population metaheuristic that scores $I_{\max}$ candidate allocations each requiring an $\mathcal{O}(K\log K)$ subcarrier water-filling costs $\mathcal{O}\!\big(I_{\max} K\log K\big)$. The proposed greedy algorithm is cheaper than both approaches and fits comfortably within a per-slot scheduling budget.

\subsubsection{Optimality of the Greedy}
Problem (P1) is NP-hard in general, but the greedy is provably optimal in a
regime that explains why it performs so well in practice.

\begin{proposition}\label{prop:opt}
Suppose on-device denoising time is negligible relative to communication ($\tau_k\!\to\!0$ for all $k$), or equivalently, that the objective counts transmission time only. Then each user's latency-minimizing rate is its smallest feasible rate, $r_k^\star=\min\{r\in\mathcal{R}: (r,m)\in\mathcal{F}_k \text{ for some } m\}$, reached with the largest number of denoising steps the energy budget allows; and given these rate-step pairs, the subcarrier allocation produced by Algorithm~\ref{alg:alloc} is globally optimal for (P1).
\end{proposition}

\begin{IEEEproof}[Proof]
The argument has two parts: first we fix each user's rate-step pair, then we show that the resulting subcarrier split is exactly the one the greedy finds.

\emph{(i) Rate-step Pair:} With $\tau_k\!\to\!0$, the latency of user $k$ at rate-step pair $(r,m)$ with $c$ subcarriers is $\ell(r)/(c B\varepsilon_k)$: it grows with the payload $\ell(r)$ and is independent of $m$. Minimizing latency therefore always prefers the smallest feasible rate $r_k^\star$, which by Eq.~\eqref{eq:feasible} stays feasible provided enough denoising steps are spent, up to the energy cap. Since denoising costs no time in this regime, the user can spend as many steps as its budget permits to meet the QoS, and reduce the number of transmitted symbols.

\emph{(ii) Subcarrier split:} Substituting $r_k^\star$ yields $T_k^\star(c)=A_k/c$ with $A_k=\ell(r_k^\star)/(B\varepsilon_k)>0$, which is strictly decreasing and strictly convex in the integer $c\ge 1$. The latency saved by giving user $k$ one additional subcarrier is given as:
\begin{equation}
\Delta_k(c)=A_k\Big(\tfrac{1}{c}-\tfrac{1}{c+1}\Big)=\frac{A_k}{c(c+1)},
\end{equation}
which strictly decreases in $c$. Each extra subcarrier helps less than the
previous one. Problem (P1) then collapses to
\begin{equation}
\min \sum_k \frac{A_k}{c_k}
\quad\text{s.t.}\quad \sum_k c_k\le N,\ \ c_k\ge 1\ \text{integer},    
\end{equation}
a separable resource-allocation problem with convex costs and diminishing marginal gains. For this classical problem, the incremental greedy, which repeatedly awards the next unit to whichever user currently offers the largest marginal gain, is globally optimal~\cite{federgruen1986greedy,7892899}. The reason is that the multiset of all achievable one-subcarrier gains $\{\Delta_k(c)\}$ is fixed in advance, and because the gains only shrink with $c$, awarding them in decreasing order (precisely what the greedy does after the mandatory first subcarrier per user) selects exactly the $N-|\mathcal{K}_{\mathrm s}|$ largest of them. Algorithm~\ref{alg:alloc} implements this rule and therefore attains the optimum of (P1).
\end{IEEEproof}

\begin{table}[t]
\caption{Main Simulation Parameters}
\label{tab:params}
\centering
\begin{tabular}{ll}
\toprule
Parameter & Value \\
\midrule
Source dataset (train/test) & DIV2K \\
Max latent channels $C_{\max}$ & $32$ \\
Supported rates $\mathcal{R}$ & $\{8,16,24,32\}$ \\
Corresponding CBR & $\{1/48,\,1/24,\,1/16,\,1/12\}$ \\
Denoising steps $\mathcal{M}$ & $\{0,1,\dots,5\}$ \\
Training SNRs (sampled) & $\{1,4,7,10,13\}$ dB \\
Profiling SNR grid $\mathcal{S}$ & $1$-$13$ dB ($1$-dB spacing) \\
Number of subcarriers $N$ & $32$ \\
Subcarrier bandwidth $B$ &  $50$ MHz\\
Mobile Computing Frequency $F_{k}$ &  $2.5$ TFLOPs\\
Optimizer / learning rate & Adam / $5\times10^{-5}$ \\
\bottomrule
\end{tabular}
\end{table}
When denoising time is not negligible ($\tau_k>0$), this guarantee can break. Each user's best-response latency $T_k^\star(c)$ in~\eqref{eq:bestresp} is now the lower envelope (pointwise minimum) of the $|\mathcal{F}_k|$ curves $\ell(r)/(cB\varepsilon_k)+m\tau_k$, one per feasible rate-step pair. Every such curve is convex and decreasing in $c$, but the pointwise minimum of convex functions need not be convex, so the marginal gains $\Delta_k(c)$ need no longer decrease monotonically. Intuitively, one extra subcarrier might barely help a user, but a second one can help a great deal, since together they free up enough spectrum for the user to jump to a cheaper choice altogether, one with a higher rate and fewer denoising steps needed. This behavior breaks the diminishing-returns property required by our optimality proof. Therefore, when denoising time is not negligible, we use it as a heuristic and evaluate its performance through experiments.

\begin{figure}[t]
        \centering
        \includegraphics[width=1\linewidth]{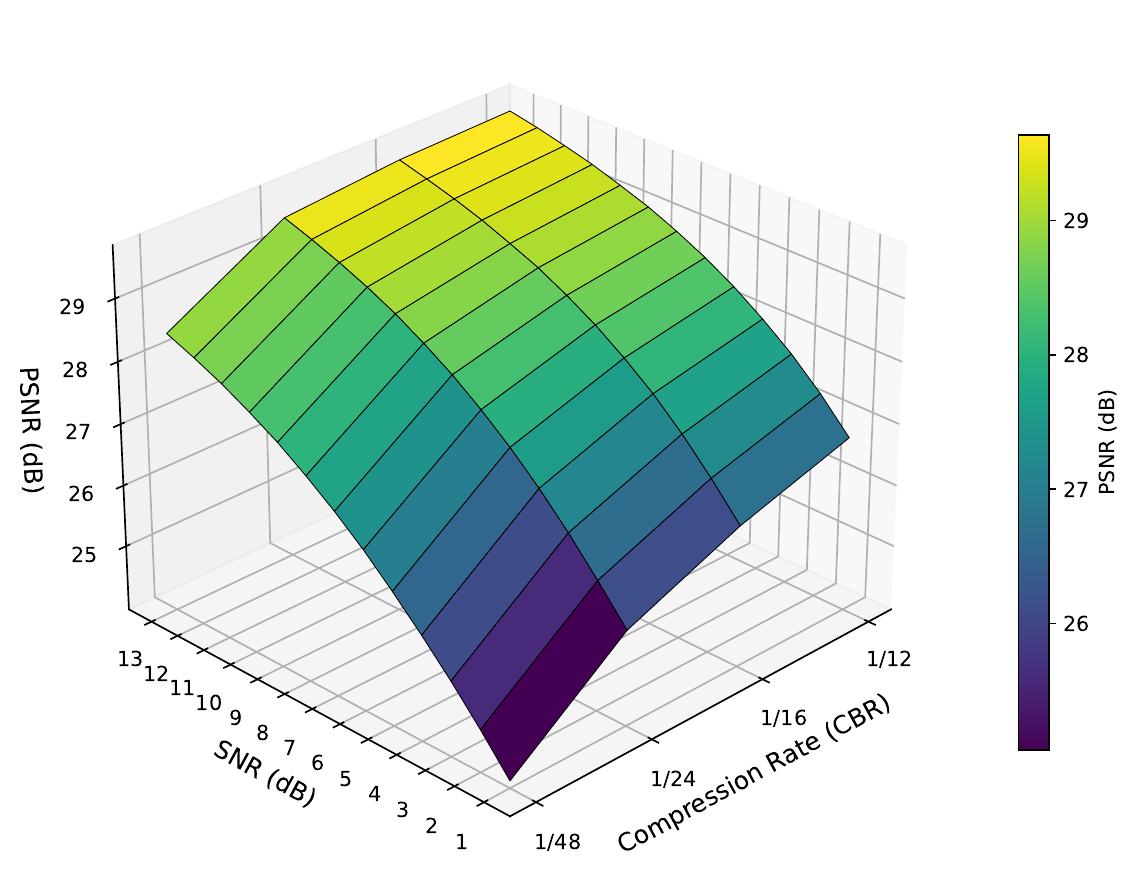}
        \caption{The average PSNR metric versus SNR and CBR.}
    \label{fig:cbr_results}
\end{figure}

\section{Simulation Results}\label{sec:results}

\subsection{Simulation Setup}
\textbf{Datasets:} The deep joint source-channel encoder/decoder and denoiser are trained on the DIV2K high-resolution image dataset with 800 images and evaluated on the remaining 100 images in the testing dataset, following common practice in DeepJSCC evaluation. 

\textbf{Baselines:} To demonstrate the effectiveness of our proposal, we compare our proposed framework with these benchmarks:
\begin{itemize}
    \item Rate Adaptation Only: This benchmark adjusts the compression rate of the transmitted signal from the BS to users in order to meet their demand, with m=$0$ fixed. 
    \item Denoising Only: This benchmark fixes the compression rate and offloads the task of meeting the QoS to local devices through a denoising process. 
    \item Maximum-rate: The benchmark blindly selects the highest available compression rate for each user. 
    \item Randomly Select from Feasible: The scenario randomly selects a pair of solutions from the feasible set. We consider this scenario to illustrate the performance of our resource allocation approach.
\end{itemize}

\textbf{Metrics:} As shown in the formulated problem, our target is to reduce the total latency, including the wireless transmission time alongside the computing time incurred by the denoising step. We adopt the PSNR as the QoS metric requested by each user, since it varies inversely with the mean squared error (MSE) between the original and reconstructed image:
\begin{equation}
\textrm{PSNR}= 10\log_{10}\frac{\textrm{MAX}^{2}}{\textrm{MSE}},
\label{PSNRreverseMSE}
\end{equation}
Here, \textrm{MAX} represents the largest attainable pixel intensity in the image \cite{10423076}, which equals 255 for images encoded with 8 bits per pixel in each color channel. In addition, we also adopt the multi-scale structural similarity index measure (MS-SSIM)~\cite{wang2003multiscale} as the metric for the image quality.

\begin{table}[t]
\centering
\caption{MS-SSIM under different channel SNRs and compression rates without denoising steps.}
\label{tab:msssim_cbr}
\renewcommand{\arraystretch}{1.15}
\setlength{\tabcolsep}{8pt}
\begin{tabular}{c|cccc}
\hline
\textbf{SNR} & \textbf{CR=1/48} & \textbf{CR=1/24} & \textbf{CR=1/16} & \textbf{CR=1/12} \\
\hline
1 dB  & 0.8008 & 0.8611 & 0.8879 & 0.9031 \\
2 dB  & 0.8294 & 0.8820 & 0.9051 & 0.9181 \\
3 dB  & 0.8533 & 0.8995 & 0.9193 & 0.9298 \\
4 dB  & 0.8731 & 0.9137 & 0.9303 & 0.9387 \\
5 dB  & 0.8894 & 0.9251 & 0.9389 & 0.9455 \\
6 dB  & 0.9031 & 0.9342 & 0.9456 & 0.9509 \\
7 dB  & 0.9142 & 0.9416 & 0.9509 & 0.9552 \\
8 dB  & 0.9234 & 0.9474 & 0.9552 & 0.9587 \\
9 dB  & 0.9308 & 0.9522 & 0.9586 & 0.9615 \\
10 dB & 0.9370 & 0.9561 & 0.9614 & 0.9638 \\
11 dB & 0.9420 & 0.9593 & 0.9637 & 0.9657 \\
12 dB & 0.9460 & 0.9619 & 0.9656 & 0.9671 \\
13 dB & 0.9494 & 0.9640 & 0.9671 & 0.9683 \\
\hline
\end{tabular}
\end{table}
\subsection{The performance of the Adaptive Codec}

As shown in Fig.~\ref{fig:cbr_results}, both PSNR and MS-SSIM increase monotonically with SNR and CBR, confirming that the single adaptive codec generalizes smoothly across all tested SNR and compression-rate configurations, without any rate- or channel- specific retraining. The gains from a longer signal are most noticeable under poor channel conditions: at $1$ dB, raising the CBR from $1/48$ to $1/12$ improves PSNR by $2.72$ dB from $24.06$ to $26.78$ dB and MS-SSIM by $0.1023$. While at $13$ dB, the same rate increase yields only a $1.27$ dB PSNR gain. This convergence at high SNR indicates the system can efficiently achieve good performance without the need for a longer signal; the channel is already clean, while at low SNR, the signal contains a large amount of noise which demands more information to achieve high performance. The same phenomenon can be observed in the MS-SSIM metric in Table~\ref{tab:msssim_cbr}: the improvement is $0.1023$ from compression rate $1/48$ to $1/12$ at $1$ dB, while this improvement value is only $0.0189$ at $13$ dB.

\subsection{Value of Receiver Computation}

\begin{figure}[t]
    \centering
    \includegraphics[width=\linewidth]{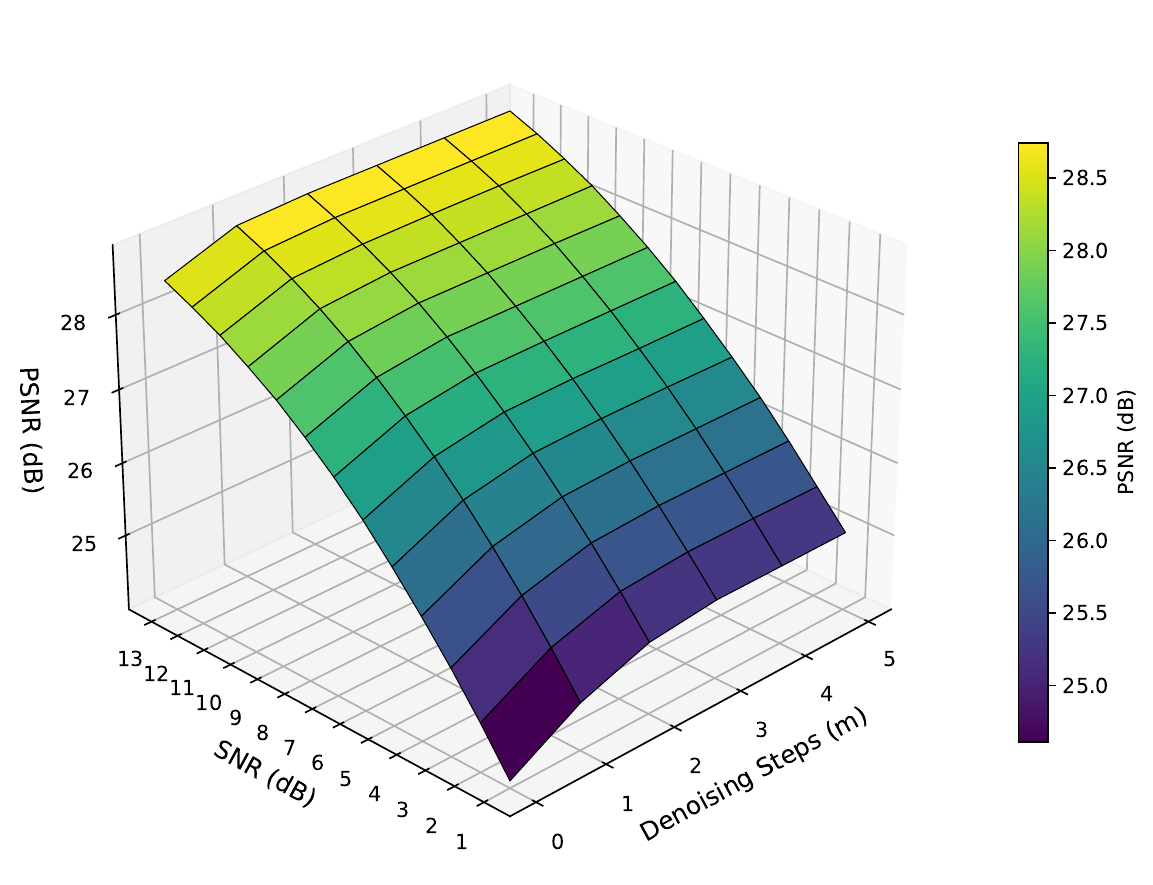}
    \caption{Effect of the denoising iteration number on reconstruction quality under different channel SNRs under the compression rate $1/48$.}
    \label{fig:denoise_results}
\end{figure}

\begin{table}[t]
\centering
\caption{MS-SSIM under different channel SNRs and denoising iterations ($\mathrm{CR}=1/48$).}
\label{tab:msssim_denoise}
\renewcommand{\arraystretch}{1.15}
\setlength{\tabcolsep}{7pt}
\scalebox{0.95}{
\begin{tabular}{c|cccccc}
\hline
\textbf{SNR} & $\mathbf{m=0}$ & $\mathbf{m=1}$ & $\mathbf{m=2}$ & $\mathbf{m=3}$ & $\mathbf{m=4}$ & $\mathbf{m=5}$ \\
\hline
1 dB  & 0.8008 & 0.8115 & 0.8375 & 0.8502 & 0.8519 & 0.8518 \\
2 dB  & 0.8294 & 0.8373 & 0.8584 & 0.8677 & 0.8685 & 0.8683 \\
3 dB  & 0.8533 & 0.8593 & 0.8763 & 0.8829 & 0.8831 & 0.8830 \\
4 dB  & 0.8731 & 0.8776 & 0.8913 & 0.8956 & 0.8958 & 0.8957 \\
5 dB  & 0.8894 & 0.8927 & 0.9039 & 0.9067 & 0.9066 & 0.9065 \\
6 dB  & 0.9031 & 0.9055 & 0.9144 & 0.9160 & 0.9159 & 0.9158 \\
7 dB  & 0.9142 & 0.9163 & 0.9232 & 0.9241 & 0.9240 & 0.9239 \\
8 dB  & 0.9234 & 0.9253 & 0.9305 & 0.9309 & 0.9309 & 0.9308 \\
9 dB  & 0.9308 & 0.9327 & 0.9366 & 0.9368 & 0.9368 & 0.9367 \\
10 dB & 0.9370 & 0.9390 & 0.9417 & 0.9418 & 0.9418 & 0.9417 \\
11 dB & 0.9420 & 0.9440 & 0.9460 & 0.9460 & 0.9460 & 0.9459 \\
12 dB & 0.9460 & 0.9481 & 0.9495 & 0.9495 & 0.9495 & 0.9494 \\
13 dB & 0.9494 & 0.9514 & 0.9524 & 0.9524 & 0.9523 & 0.9523 \\
\hline
\end{tabular}}
\end{table}
Fig.~\ref{fig:denoise_results} shows the relation between the reconstruction quality and the number of denoising steps $m$ and the channel condition reflected by the SNRs, confirming that the receiver-side refinement consistently improves fidelity. However, the gains follow a clear pattern, saturating after approximately $m$= $3$ steps: at $1$ dB, PSNR improves by $0.98$ dB and MS-SSIM by $0.0511$ as $m$ increases from $0$ to $3$. This empirical result indicates that the denoising module can actually improve the performance but with limited additional gains. At higher SNRs, both the achievable gain and the number of steps needed to reach it shrink substantially. Specifically, at 13 dB, PSNR improves by only 0.37 dB, and the improvement saturates by $m$ = 2. This confirms that on-device denoising shares the same property as the increasing signal length approach and is most valuable precisely when the channel condition is poor.
\begin{table*}[t]
\centering
\caption{PSNR (dB) and MS-SSIM under different channel SNRs, compression rates (CRs), and denoising iterations.}
\label{tab:master_psnr_msssim}
\renewcommand{\arraystretch}{1.3}
\setlength{\tabcolsep}{3.0pt}
\scriptsize
\begin{tabular}{cc|cccccc|cccccc|cccccc}
\hline
& &
\multicolumn{6}{c|}{\textbf{CR = 1/24}} &
\multicolumn{6}{c|}{\textbf{CR = 1/16}} &
\multicolumn{6}{c}{\textbf{CR = 1/12}} \\
\cline{3-20}
\textbf{SNR} & \textbf{Metric}
& $m_0$ &$m_1$ &$m_2$ &$m_3$ &$m_4$ &$m_5$
& $m_0$ &$m_1$ &$m_2$ &$m_3$ &$m_4$ &$m_5$
& $m_0$ &$m_1$ &$m_2$ &$m_3$ &$m_4$ &$m_5$
\\
\hline
\multirow{2}{*}{1} & PSNR &
25.49&26.17&26.44&26.53&26.52&26.52&26.24&26.95&27.21&27.28&27.28&27.27&26.78&27.40&27.70&27.76&27.76&27.75\\
& MS-SSIM &
0.8611&0.8743&0.8883&0.8969&0.8977&0.8976&0.8879&0.8989&0.9107&0.9169&0.9174&0.9172&0.9031&0.9103&0.9231&0.9280&0.9282&0.9281\\
\cline{2-20}
\multirow{2}{*}{2} & PSNR &
26.05&26.67&26.91&26.97&26.97&26.96&26.78&27.40&27.63&27.68&27.67&27.67&27.28&27.82&28.09&28.13&28.12&28.12\\
& MS-SSIM &
0.8820&0.8918&0.9031&0.9091&0.9095&0.9094&0.9051&0.9124&0.9221&0.9263&0.9265&0.9263&0.9181&0.9221&0.9325&0.9358&0.9358&0.9357\\
\cline{2-20}
\multirow{2}{*}{3} & PSNR &
26.58&27.13&27.34&27.38&27.38&27.38&27.27&27.81&28.02&28.05&28.04&28.04&27.71&28.20&28.42&28.45&28.45&28.45\\
& MS-SSIM &
0.8995&0.9061&0.9157&0.9196&0.9197&0.9196&0.9193&0.9237&0.9316&0.9342&0.9343&0.9342&0.9298&0.9319&0.9402&0.9423&0.9423&0.9422\\
\cline{2-20}
\multirow{2}{*}{4} & PSNR &
27.06&27.55&27.74&27.77&27.76&27.76&27.70&28.18&28.36&28.38&28.38&28.38&28.07&28.54&28.73&28.76&28.75&28.75\\
& MS-SSIM &
0.9137&0.9183&0.9260&0.9283&0.9283&0.9284&0.9303&0.9330&0.9394&0.9410&0.9410&0.9409&0.9387&0.9400&0.9465&0.9478&0.9477&0.9477\\
\cline{2-20}
\multirow{2}{*}{5} & PSNR &
27.49&27.93&28.10&28.12&28.11&28.11&28.06&28.51&28.67&28.69&28.68&28.68&28.39&28.84&29.00&29.02&29.02&29.01\\
& MS-SSIM &
0.9251&0.9282&0.9343&0.9358&0.9358&0.9357&0.9389&0.9407&0.9456&0.9466&0.9465&0.9465&0.9455&0.9466&0.9516&0.9524&0.9523&0.9523\\
\cline{2-20}
\multirow{2}{*}{6} & PSNR &
27.87&28.28&28.43&28.45&28.44&28.44&28.38&28.80&28.94&28.96&28.96&28.95&28.66&29.10&29.24&29.25&29.25&29.25\\
& MS-SSIM &
0.9342&0.9365&0.9413&0.9421&0.9421&0.9420&0.9456&0.9470&0.9508&0.9513&0.9513&0.9512&0.9509&0.9520&0.9559&0.9562&0.9562&0.9561\\
\cline{2-20}
\multirow{2}{*}{7} & PSNR &
28.21&28.60&28.73&28.74&28.73&28.73&28.65&29.06&29.19&29.20&29.20&29.20&28.89&29.32&29.44&29.45&29.46&29.46\\
& MS-SSIM &
0.9416&0.9435&0.9470&0.9475&0.9474&0.9474&0.9509&0.9522&0.9551&0.9553&0.9553&0.9553&0.9552&0.9564&0.9592&0.9594&0.9594&0.9593\\
\cline{2-20}
\multirow{2}{*}{8} & PSNR &
28.50&28.89&28.99&29.00&29.00&29.00&28.89&29.29&29.40&29.41&29.41&29.41&29.09&29.52&29.62&29.63&29.63&29.63\\
& MS-SSIM &
0.9474&0.9492&0.9518&0.9521&0.9520&0.9520&0.9552&0.9565&0.9586&0.9587&0.9587&0.9587&0.9587&0.9600&0.9620&0.9621&0.9620&0.9620\\
\cline{2-20}
\multirow{2}{*}{9} & PSNR &
28.76&29.14&29.23&29.23&29.23&29.23&29.08&29.49&29.58&29.59&29.59&29.59&29.25&29.68&29.77&29.78&29.78&29.78\\
& MS-SSIM &
0.9522&0.9538&0.9558&0.9559&0.9559&0.9558&0.9586&0.9599&0.9615&0.9615&0.9615&0.9614&0.9615&0.9628&0.9643&0.9643&0.9643&0.9642\\
\cline{2-20}
\multirow{2}{*}{10} & PSNR &
28.98&29.35&29.43&29.44&29.43&29.43&29.26&29.66&29.74&29.75&29.75&29.75&29.40&29.82&29.90&29.91&29.91&29.91\\
& MS-SSIM &
0.9561&0.9577&0.9591&0.9592&0.9591&0.9591&0.9614&0.9627&0.9638&0.9638&0.9638&0.9638&0.9638&0.9651&0.9661&0.9661&0.9661&0.9661\\
\cline{2-20}
\multirow{2}{*}{11} & PSNR &
29.17&29.55&29.61&29.61&29.62&29.61&29.41&29.80&29.87&29.88&29.88&29.88&29.52&29.94&30.00&30.01&30.01&30.01\\
& MS-SSIM &
0.9593&0.9608&0.9619&0.9618&0.9618&0.9618&0.9637&0.9650&0.9658&0.9658&0.9657&0.9657&0.9657&0.9669&0.9676&0.9676&0.9676&0.9676\\
\cline{2-20}
\multirow{2}{*}{12} & PSNR &
29.34&29.71&29.76&29.77&29.76&29.76&29.54&29.93&29.99&29.99&29.99&29.99&29.63&30.04&30.09&30.10&30.10&30.10\\
& MS-SSIM &
0.9619&0.9634&0.9641&0.9641&0.9641&0.9640&0.9656&0.9668&0.9674&0.9673&0.9673&0.9673&0.9671&0.9683&0.9688&0.9688&0.9688&0.9688\\
\cline{2-20}
\multirow{2}{*}{13} & PSNR &
29.49&29.84&29.89&29.89&29.89&29.89&29.66&30.03&30.08&30.08&30.08&30.09&29.72&30.11&30.16&30.17&30.17&30.17\\
& MS-SSIM &
0.9640&0.9654&0.9660&0.9659&0.9659&0.9659&0.9671&0.9682&0.9686&0.9686&0.9686&0.9686&0.9683&0.9695&0.9698&0.9698&0.9698&0.9698\\
\hline
\end{tabular}
\end{table*}

The above results are obtained for the compression rate $1/48$, while Table~\ref{tab:master_psnr_msssim} captures a complete performance behavior of our denoising module under different compression rates, confirming that the trends observed at CR $=1/48$ generalize across the operating range. At every compression rate, PSNR and MS-SSIM typically increase with SNR and rate, and denoising again exhibits limited gains in improving performance once the denoising steps exceed three. It is worth noting that the absolute denoising gain remains largest at low SNR: at SNR $=1$ dB, steps $m_0\rightarrow m_3$ improve PSNR by about 1.0 dB at CR $=1/24$ and $1/16$, and 0.98 dB at CR $=1/12$. This performance gain under the harsh wireless condition is actually the target we want to achieve, since when the channel condition is good, the user QoS is already satisfied without the need for signal denoising.

\begin{figure}[t]
\centering
\includegraphics[width=0.5\textwidth]{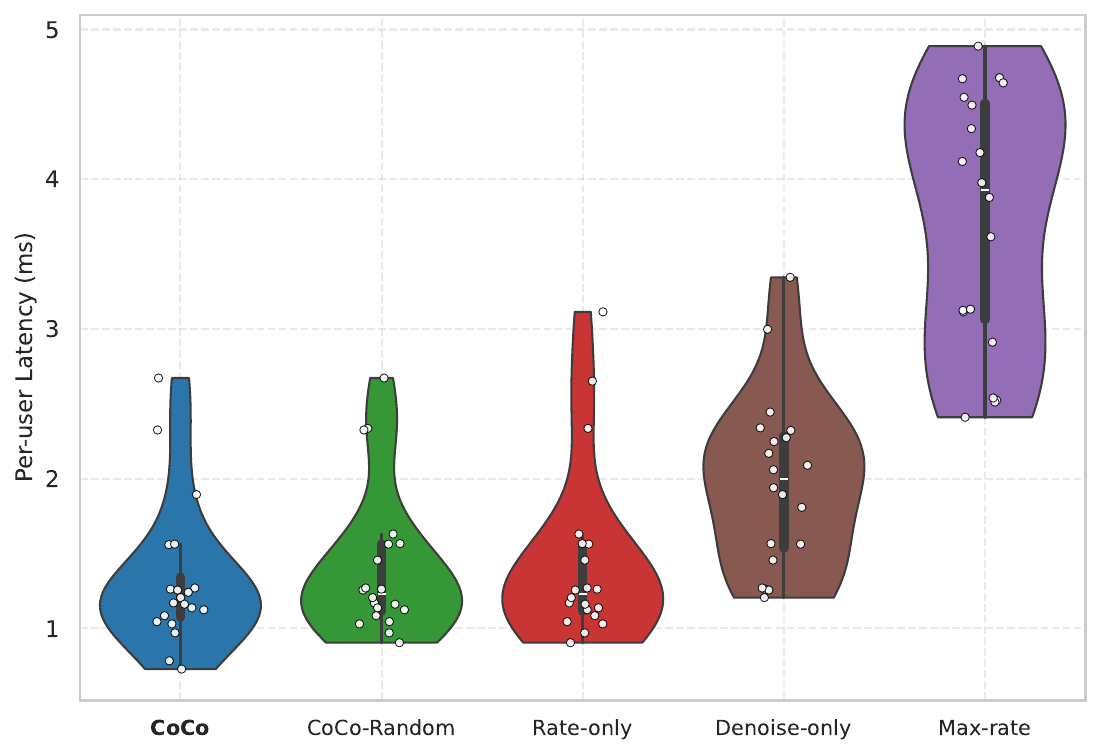}
\caption{Per-user latency distribution under CoCo versus four benchmarks ($K=20$ users, $N=32$ subcarriers).}
\label{figcocobenchamrks}
\end{figure}

\subsection{The efficiency of the Compression-Computation framework}

Fig.~\ref{figcocobenchamrks} compares the per-user latency distributions of the proposed CoCo allocation against four benchmarks: rate-only, denoise-only, max-rate, and CoCo-Random, which selects a set from our constructed feasible set. CoCo achieves the lowest average latency with $1.324$ ms while meeting every user's quality requirement. CoCo-Random, which also focuses on improving the quality by computing and communicating but randomly selects a feasible set of compression rate and denoising, which achieves the average latency at $1.408$ ms, confirming that most of CoCo's gain comes from jointly optimizing rate and denoising steps rather than from subcarrier allocation alone. The rate-only and denoise-only secure $1.446$ ms and $1.975$ ms in latency, respectively. This result shows the inefficiency of exploiting a single approach, and these approaches also leave two users' demands unmet. Always-max-rate performs worst by far, at $3.715$ ms, showing that indiscriminately maximizing quality wastes communication bandwidth. These results demonstrate that co-adapting compression and computation is essential for latency-efficient, quality-guaranteed multi-user delivery.

\begin{figure*}[t]
\centering
\includegraphics[width=1\textwidth]{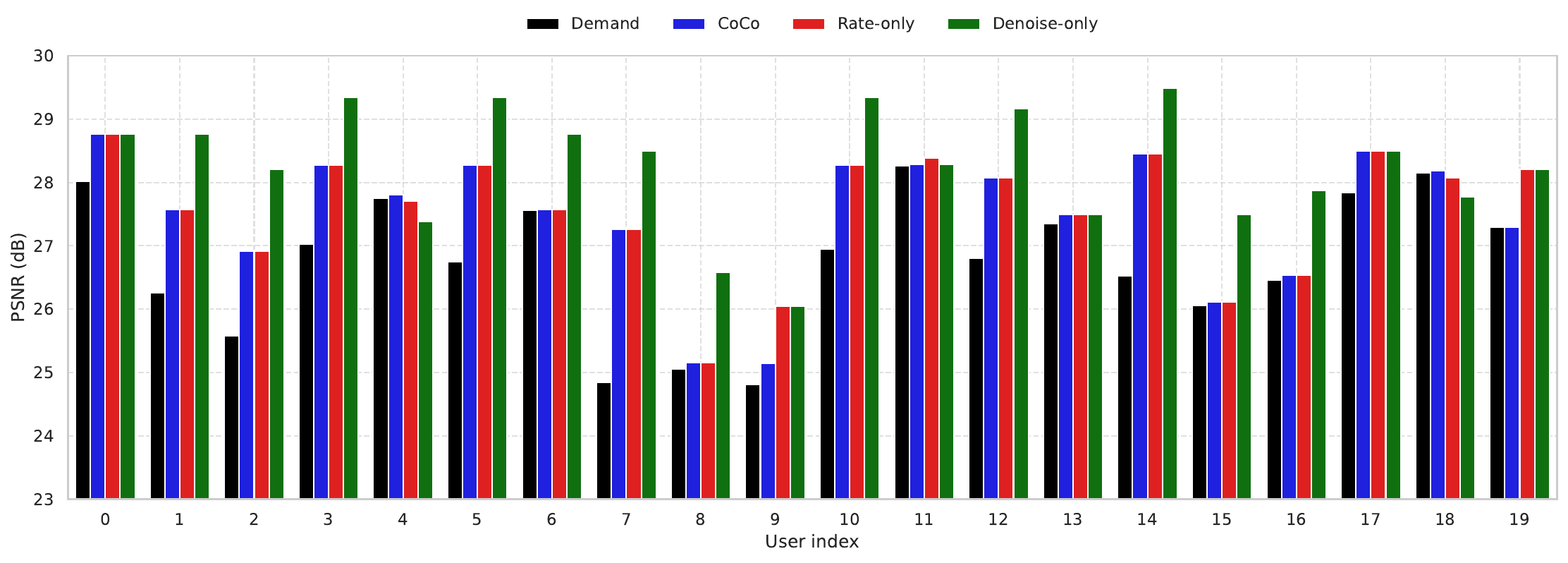}
\caption{Per-user PSNR demand versus achieved quality under CoCo, Rate-only, and Denoise-only. CoCo meets every user's target with minimal overshoot, while the single-axis baselines both under-serve some users and over-provision others.}
\label{figdemandsandachieve}
\end{figure*}

Fig.~\ref{figdemandsandachieve} compares each user's PSNR demand against the achieved quality under CoCo, Rate-only, and Denoise-only. CoCo meets every user's target while keeping overshoot modest, at $0.73$ dB on average, since it can freely trade compression rate against denoising steps to facilitate the requirement. However, the rate-only and denoise-only both fail two users whose targets cannot be reached. It is worth noticing that the denoise-only overshoots by $1.50$ dB on average, which indicates that it allocates more resources for users with low expectation while failing to assist users with high demand and under low SNR. This scenario can be explained by the fixed compression rate of the denoising-only benchmark. The fixed compression rate also leads to higher latency, and the system cannot gain from the overshooting scenario but can lose profit due to under-serving clients. The proposed CoCo provides a flexible mechanism to balance compression and computation capacity for user demand while keeping latency low.

\begin{table}[t]
\caption{Selected Compression Rates, Denoising Steps, and On-Device Denoising Energy}
\label{tab:energy}
\centering
\begin{tabular}{lccccc}
\toprule
 & \multicolumn{3}{c}{Selected operating points} & \multicolumn{2}{c}{Energy} \\
\cmidrule(lr){2-4}\cmidrule(lr){5-6}
Policy & Rates $r_k$ & $\bar{m}$ & Latency & Total & Per user \\
 & (\#users) & (steps) & (ms) & (mJ) & (mJ) \\
\midrule
CoCo          & $14/4/2/0$  & $0.25$ & $\mathbf{1.324}$ & $40.37$ & $2.02$ \\
CoCo-Random   & $12/5/3/0$  & $0.10$ & $1.408$ & $16.15$ & $0.81$ \\
Rate-only     & $12/5/1/2$  & $0$    & $1.446$ & $0$     & $0$ \\
Denoise-only  & $0/20/0/0$  & $0.35$ & $1.975$ & $56.52$ & $2.83$ \\
Always-max    & $0/0/0/20$  & $0$    & $3.715$ & $0$     & $0$ \\
\bottomrule
\end{tabular}
\\[2pt]
\footnotesize ``Rates'' lists the number of users served at $1/48$, $1/24$, $1/16$,
and $1/12$, respectively; $\bar{m}$ is the average number of denoising steps.
\end{table}

Table~\ref{tab:energy} compares the selected compression rates, denoising steps, latency, and energy consumption of different policies. CoCo mainly relies on high compression, assigning $14$ out of $20$ users to the $1/48$ rate, and only uses denoising when it is needed. As a result, the average number of denoising steps is only $0.25$, leading to a total energy consumption of $40.37$~mJ ($2.02$~mJ per user). In contrast, Denoise-only depends more on computation, increasing the average denoising steps to $0.35$ and the energy consumption to $56.52$~mJ, while still having higher latency ($1.975$~ms compared to $1.324$~ms).

\begin{figure}[t]
\centering
\includegraphics[width=0.5\textwidth]{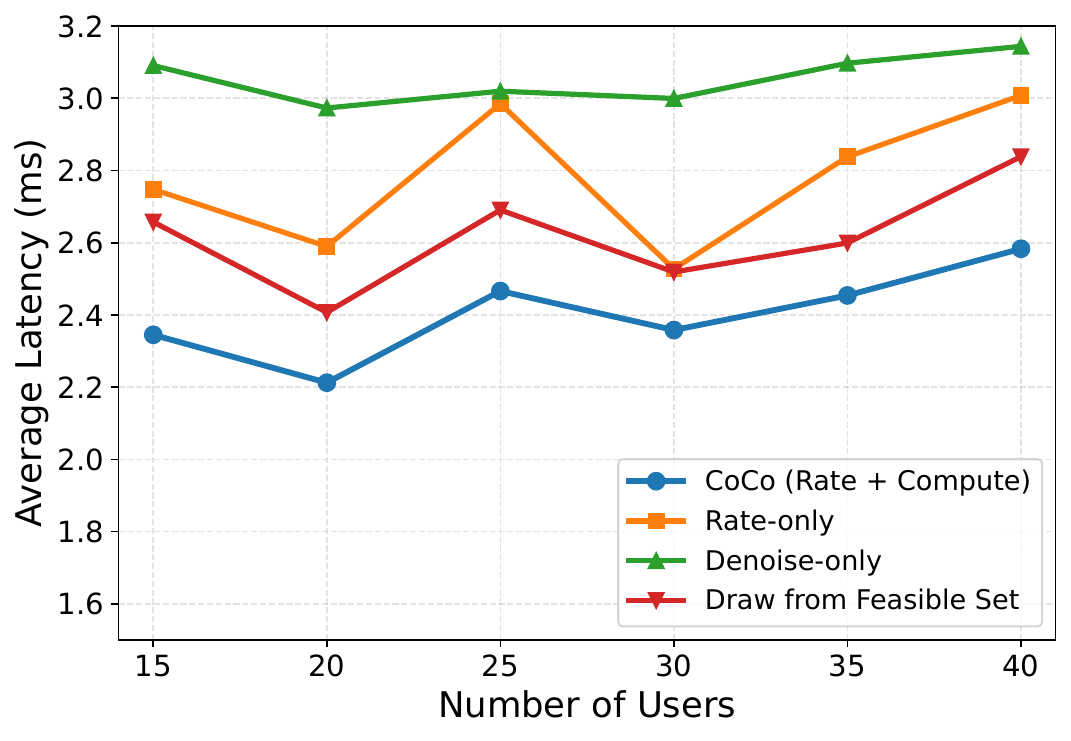}
\caption{The performance of the proposal compared with other approaches when the number of users varies and the number of sub-carriers is equal to the number of users.}
\label{userchanges}

\end{figure}

\begin{figure*}[t]
\centering
\includegraphics[width=1\linewidth]{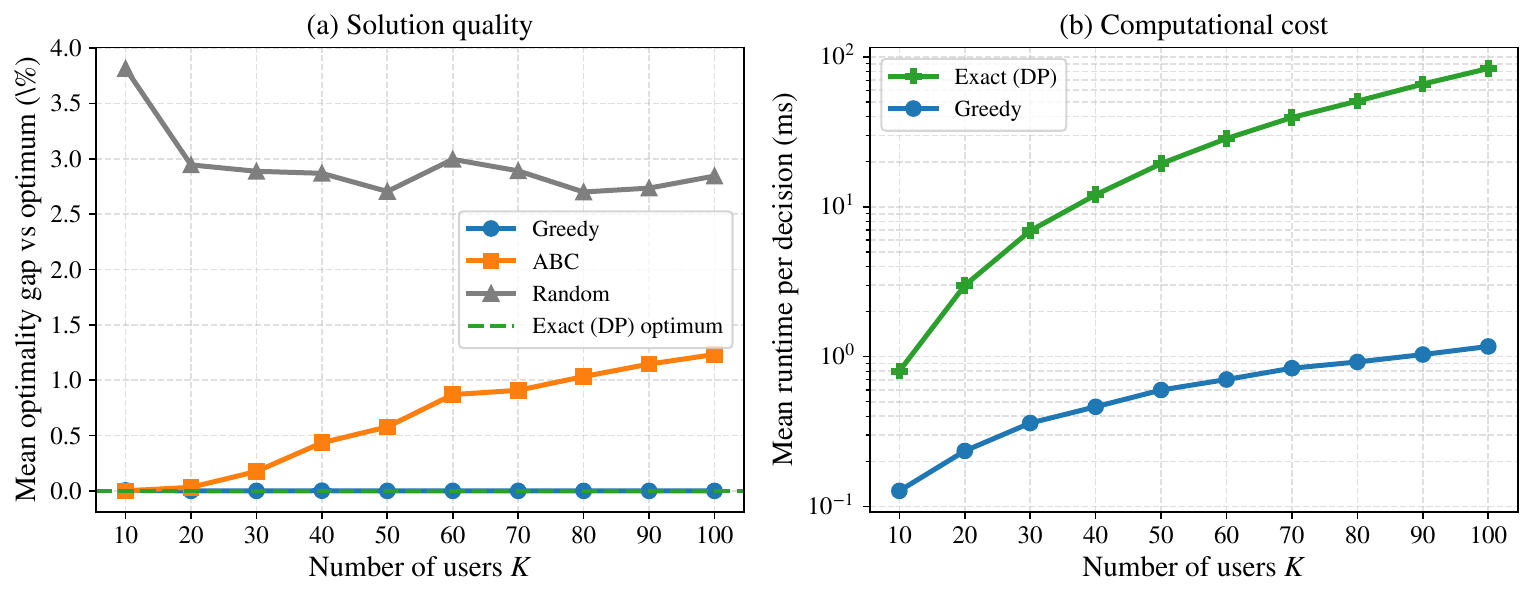}
\caption{Allocation quality and cost when the number of users $K$ increases. (a) Optimality gap against the exact optimum: the greedy tracks it almost exactly at every $K$, while ABC drifts further away as the search space grows. (b) The runtime when the number of decisions is increase: the greedy is cheaper than the exact solver.}
\label{fig:alloc_optimality}
\end{figure*}

\subsection{Change in Number of Users}
Up to this point, we have considered $20$ users with 32 sub-carriers, where a single user can have more than one sub-carrier to reduce the communication time. In this section, we present simulation results for different numbers of users, where each user can have at most one sub-carrier, which is typically the case in real-world environments. As shown in Fig.~\ref{userchanges}, we provide the latency of different adaptation strategies as the number of users increases from 15 to 40. The proposed CoCo obtains the lowest latency across all network sizes, demonstrating the advantage of jointly optimizing network communication bandwidth and local computation capacity. CoCo outperforms the rate-only approach by approximately $8$\%-$18$\%, indicating that communication resources alone are not the optimal solution to meet the user demands, and increasing the number of transmitted symbols can increase the latency. The denoising-only scheme incurs even higher latency since the number of sub-carriers is limited, resulting in unnecessary communication overhead for some users while underperforming for others with higher expectations. Finally, the random feasible allocation performs better than the single-dimension baselines but performs worse than the optimal pair, highlighting the importance of selecting optimal rate and denoising steps. Finally, the latency of our proposal remains stable as the number of users increases, which suggests the scalability of the method.

\subsection{Optimality Gap and Complexity of the Greedy Allocation}\label{sec:results-alloc}
Section~\ref{sec:alloc} shows the greedy allocator is provably optimal only when denoising time is negligible, and treats the general case as a heuristic to be tested empirically. Here, we evaluate how much Algorithm~\ref{alg:alloc} actually falls short of the true optimum once denoising time is not negligible, and the cost of it. We compare it against three references on problem instances drawn from the same distribution used above and scored with the measured lookup table $Q[r,m,\rho]$:
\begin{itemize}
  \item \emph{Exact optimum}: for a fixed number of subcarriers per user, each user's best rate-step choice no longer depends on any other user's, so~(P2) reduces to choosing each $c_k$ to minimize $\sum_k T_k^\star(c_k)$ subject to $\sum_k c_k\le N$. A dynamic program over the shared subcarrier budget solves this exactly in $\mathcal{O}(KN^2)$ time, for any $K$ or $N$~\cite{bellman1966dynamic}.
  \item \emph{Artificial Bee Colony (ABC)~\cite{karaboga2007powerful}}: a honey-bee swarm metaheuristic that searches the same operating-point space, with employed, onlooker, and scout bees, under a fixed evaluation budget.
  \item \emph{Random}: the Randomly Select from Feasible policy of Section~\ref{sec:results}, i.e., one feasible operating point per user drawn at random, with no optimization at all.
\end{itemize}
All three approaches use the identical exact subcarrier water-filling for the inner allocation, so that the difference in total latency comes entirely from how well each one searches the per-user rate and step choices. Across the $500$ instances spanning $K=10$ to $100$, the greedy matches the exact optimum in $494$ of them; the six exceptions differ by at most $0.24\%$. At this fidelity, the heuristic behaves as if it were an exact solver. ABC starts out just as strong at small $K$, but its fixed search budget cannot keep pace as the rate-step space grows due to the number of users increasing, and its gap widens from $0.03\%$ at $K=20$ to $1.23\%$ at $K=100$. The random baseline never competes: it stays $2.7$--$3.8\%$ above optimal regardless of $K$, the cost of not optimizing at all, as shown in Fig.~\ref{fig:alloc_optimality}.

\section{Conclusion}\label{sec:conclusion}
In this paper, we proposed CoCo, an energy-aware compression-computation co-adaptation framework for latency minimization in multi-user semantic communications. The base station has to serve users with heterogeneous signal-to-noise ratios, quality targets, and device energy budgets over a shared spectrum. To facilitate the user's requirement, the base station decides whether to assist it with a longer transmitted signal or offload the task to a local device to perform on-device denoising. Based on this, we have formulated an optimization problem to minimize the communication time of all the users in the network subject to per-user quality of service and energy constraints, where the control variables are the compression rates, the number of local denoising steps, and the number of allocated sub-carriers. Our training approach for the encoder, decoder, and denoiser facilitates the different modes with compression rate and denoising steps without the need to retrain. Simulation results demonstrated that the denoiser yields the largest quality gains at low rate and low SNR, when the user needs it the most, and the increase in the number of transmitted signals also improves performance under low SNR. Two approaches can complement each other instead of a trade-off, and our proposed co-adapting compression and computation effectively optimizes the communication bandwidth resource of the network and the computing capacity of users while reducing the total delivery latency and facilitating the QoS for each individual user. 

\bibliographystyle{IEEEtran}
\bibliography{mybib}

\end{document}